\documentclass{amsart}

\usepackage{amssymb, amsmath, amsthm}
\usepackage{wasysym}
\usepackage{multicol}
\usepackage{algorithm}
\usepackage{algorithmic}
\usepackage{tikz}
\usetikzlibrary{arrows}

\newtheorem{theorem}{Theorem}
\newtheorem{lemma}[theorem]{Lemma}

\theoremstyle{remark}
\newtheorem{remark}[theorem]{Remark}
\newtheorem{fact}[theorem]{Fact}
\newtheorem*{example}{Example}
\newtheorem*{remark*}{Remark}

\newcommand{\vek}{\mathcal D}
\newcommand{\ep}{\varepsilon}
\newcommand{\floor}[1]{\left\lfloor #1 \right\rfloor}

\newcommand{\E}{{\mathcal E}}

\newcommand{\bb}{\mathbf b}

\newcommand{\f}{\mathbf f}

\newcommand{\nul}{\mathbf 0}

\newcommand{\NN}{\mathbb N}
\newcommand{\sa}{s}
\newcommand{\sbb}{s'}
\newcommand{\scc}{s''}
\newcommand{\daa}{d}
\newcommand{\dbb}{d'}
\newcommand{\dcc}{d''}
\newcommand{\comment}[1]{}

\newcommand{\gmod}{\succ}
\newcommand{\gmodeq}{\succeq}

\begin{document}

\title{Abelian powers in paper-folding words}

\author{\v St\v ep\'an Holub}

\address{Department of Algebra, Charles University, Sokolovsk\'a 83, 175 86 Praha, Czech Republic}
\email{holub@karlin.mff.cuni.cz}
\urladdr{http://www.karlin.mff.cuni.cz/~holub/}

\begin{abstract}
We show that paper-folding words contain arbitrarily large abelian powers.
\end{abstract}

\maketitle
\section{Introduction}
Study of abelian powers in infinite words dates back to Erd\H{o}s's question whether there is an infinite word avoiding abelian squares (\cite{erdos}). Abelian patterns, their presence or avoidability in infinite words, are a natural generalization of analogous questions for ordinary patterns. Both variants are amply studied.  Thue's famous square-free word over three letters (see \cite{Thue1906}) has its counterpart in Ker\"anen's construction of an abelian square-free word (\cite{veikko}). In \cite{Currie}, it is shown that all long abelian patterns are avoidable over two letters.

Paper-folding words are (infinite) words that can be represented by repeated folding of a paper strip. Some interesting properties of these words and further references can be found in \cite{Allouche1,Allouche2}, see also \cite{narad}.

During the workshop \emph{Outstanding Challenges in Combinatorics on Words} at BIRS, James Currie raised the question (asked already in 2007 by Manuel Silva, \cite{pers}) whether paper-folding words contain arbitrarily large abelian powers. In this paper, we answer the question positively for all paper-folding words.

\section{Preliminaries}
Paper-folding words can be defined in several equivalent ways. They are binary words that can be represented as limits of an infinite process of folding a strip of paper. The process is governed by another binary word $\bb=b_0b_1b_2\cdots $ called the \emph{sequence of instructions} that tells whether the corresponding fold should be a hill or a valley. We shall follow the notation from \cite{jeff} and use the binary alphabet $\{1,-1\}$ for both the sequence of instructions and the resulting paper-folding word $\f=f_1f_2f_3\cdots$. Note that the sequence of instructions is indexed starting from $0$, while the paper-folding word starting from $1$, which is justified by the following defining formula:
$$f_i=(-1)^jb_k, \quad\quad \text{where $i=2^k(2j+1)$}.$$
An equivalent formulation is that 
\[f_i=-1 \quad\quad \text{iff} \quad\quad i\equiv (2+b_k)\cdot 2^k \mod 2^{k+2}. \]
If $i\equiv (2+b_k)\cdot 2^k \mod 2^{k+2}$, then we say that $f_i$ is \emph{$-1$ of order $k$}.  The \emph{regular} paper-folding word is defined by $b_i=1$ for all $i$.

A word $w$ is said to be an \emph{abelian $m$-power} if $w=w_1w_2\cdots w_m$, and the word  $w_j$ can be obtained from $w_i$ by permutation of its letters for each $i,j\in \{1,2,\dots,m\}$. In other words, $|w_i|_a=|w_j|_a$ for each letter $a$, where $|u|_a$ denotes the number of occurrences of the letter $a$ in $u$. We say that $\f$ \emph{contains an abelian $m$-power} if $f_if_{i+1}\cdots f_j$ is an ablian $m$-power for some $0<i\leq j$.

An important ingredient of our considerations will be inequalities modulo a given number (in fact, always a power of two). This requires some clarification which should prevent possible confusion. By $(a \mod n)$ we denote the unique integer in $\{0,1,\dots,n-1\}\cap (a+n \cdot \mathbb Z)$. We write \[a\gmod b \mod n\] if and only if 
\[(a \mod n) > (b \mod n)\] in the standard integer order. Similarly we define $a\gmodeq b \mod n$.

Note the following fact.
\begin{fact}\label{fact}
Let $a,b,c,d \in \mathbb Z$. If 
 \[(a \mod n)+c,\, (b\mod n)+d \in \{0,1,\dots,n-1\},\]
then
\begin{align*}
	a+c&\gmod b+d \mod n\quad \text{if and only if} \quad (a \mod n)+c>(b\mod n)+d,\\
	a+c&\gmodeq b+d \mod n\quad \text{if and only if} \quad (a \mod n)+c\geq(b\mod n)+d.
\end{align*}
\end{fact}

Note also that while $a\gmod b \mod n$ always implies $a\gmodeq b+1 \mod n$, the inverse implication never holds if $b+1\equiv 0 \mod n$.

\section{The problem}
We are given a paper folding word $\f$ defined by a sequence of instructions $\bb$. Our task is to find, for each $m$, numbers $s\geq 0$ and $d\geq 1$ such that the word
\[w_0w_2\dots w_{m-1},\]
where $w_j$, $j=0,1,\dots,m-1$, is defined by
\[w_j=f_{s+jd+1}f_{s+jd+2}\cdots f_{s+(j+1)d},\]
is na abelian $m$-power.
That is, 
\begin{align}\label{vektor}
	\left(|w_0|_{-1},|w_2|_{-1},\dots,|w_{m-1}|_{-1}\right)
\end{align}
is a constant vector.

For $b\in \{-1,1\}$, $k\geq 0$ and $\ell,n\geq 1$, denote
\[
D_{k,b}(\ell,n):=\#\{i\ |\ \ell< i\leq n,\, i\equiv (2+b)\cdot 2^k\mod 2^{k+2}\}.
\]
Then we have
\[\left|f_{\ell+1}f_{\ell+2}\cdots f_{n}\right|_{-1}=\sum_{k=0}^\infty D_{k,b_k}(\ell,n).\]
Note that $D_{k,b_k}(\ell,n)$ counts the number of $-1$s of order $k$ present in the  word $f_{\ell+1}f_{\ell+2}\cdots f_{n}$. This ``stratification'' is very useful, since $D_{k,b}(\ell,n)$  can be well estimated from the length of the interval. 
\begin{lemma}
\[D_{k,b}(\ell,n)=
\floor{\frac{n-\ell}{2^{k+2}}}
+\ep_{k,b}(\ell,n), \]
where 
\begin{align}\label{plusone}
	\ep_{k,b}(\ell,n)=\left\{ 
	\begin{array}{cl}
		1 & \text{if} \quad n-\ell\gmod(2+b)\cdot 2^k-(\ell+1) \mod 2^{k+2},\\[2ex]
		0 & \text{otherwise.} 
	\end{array}
	\right.
\end{align}
\end{lemma}
\begin{proof}
Let
\[ q=\floor{\frac{n-\ell}{2^{k+2}}}\quad\quad \text{and} \quad\quad r\equiv n-\ell \mod 2^{k+2} \]
Divide the interval $(\ell,n]$ of integers into $q+1$ subintervals
\begin{align*}
	&[\ell+j\cdot 2^{k+2}+1,\ell+(j+1)\cdot 2^{k+2}],\quad j=0,1,\dots,q-1, \\
	&[\ell+q\cdot 2^{k+2}+1,\ell+q\cdot 2^{k+2}+r]. 
\end{align*}
First $q$ intervals have length $2^{k+2}$ whence each of them contains exactly one number equal to $(2+b)\cdot 2^k
 \mod 2^{k+2}$. The last interval, possibly empty, may or may not contain such a number. One readily verifies that the value 
of $\ep_{k,b}(\ell,n)$ defined by \eqref{plusone} indicates the presence of that additional $(2+b)\cdot 2^k
 \mod 2^{k+2}$.

\end{proof}

\begin{example}
An interval $(\ell,\ell+6]$ of length six is supposed to contain one number $3 \mod 4$ and no number $6 \mod 8$. However,
if $\ell\mod 4$ is $1$ or $2$, then the interval contains two numbers $3 \mod 4$ and $\ep_{0,1}(\ell,\ell+6)=1$. Moreover, if $\ell \mod 8$ is not $6$ or $7$ then the interval contains a number $6 \mod 8$ and $\ep_{1,1}(\ell,\ell+6)=1$ (see Figure \ref{kulicky}).
\end{example}

\begin{figure}
\begin{center}
\begin{tikzpicture}
\def\roz{-0.5}
\def\rozs{0.7}
\def\slaa{-3*\rozs}
\def\slcc{-0.8*\rozs}
\def\slbb{-2*\rozs}
\node at (\slbb,-\roz){$\ep_{0,1}$};
\node at (\slcc,-\roz){$\ep_{1,1}$};
\node at (\slaa,-\roz+0.1){$\ell$};
\draw (\slaa-0.2,-\roz-0.2)--(\slcc+0.3,-\roz-0.2);
\foreach \x/\y/\z in {
0/0/1,
1/1/1,
2/1/1,
3/0/1,
4/0/1,
5/1/1,
6/1/0,
7/0/0,
8/0/1
}
{
\node at (\slaa,\x*\roz) {$\x$};
\node at (\slbb,\x*\roz) {$\y$};
\node at (\slcc,\x*\roz) {$\z$};
}
\node at (1*\rozs,0) {$1$};
\foreach \x in {0,1,2}
{\node[circle,fill,color=blue!30,inner sep=0.4em] at (3*\rozs,\x*\roz){};
\node at (3*\rozs,\x*\roz){$3$};
}
\foreach \x in {0,1}
{\node at (2*\rozs,\x*\roz){$2$};
}
\foreach \x in {0,1,2,3}
{\node at (4*\rozs,\x*\roz){$4$};
}
\foreach \x in {0,1,2,3,4}
{\node at (5*\rozs,\x*\roz){$5$};
}
\foreach \x in {0,1,2,3,4,5}
{\node[circle,draw=green,very thick,fill=green!30,inner sep=0.4em] at (6*\rozs,\x*\roz){};
\node at (6*\rozs,\x*\roz){$6$};
}
\foreach \x in {1,2,3,4,5,6}
{\node[circle,fill,color=blue!30,inner sep=0.4em] at (7*\rozs,\x*\roz){};
\node at (7*\rozs,\x*\roz){$7$};
}
\foreach \x in {2,3,...,7}
{\node at (8*\rozs,\x*\roz){$8$};
}
\foreach \x in {3,4,...,8}
{\node at (9*\rozs,\x*\roz){$9$};
}
\foreach \x in {4,5,...,8}
{\node at (10*\rozs,\x*\roz){$10$};
}
\foreach \x in {5,6,7,8}
{\node[circle,fill,color=blue!30,inner sep=0.45em] at (11*\rozs,\x*\roz){};
\node at (11*\rozs,\x*\roz){$11$};
}
\foreach \x in {6,7,8}
{\node at (12*\rozs,\x*\roz){$12$};
}
\foreach \x in {7,8}
{\node at (13*\rozs,\x*\roz){$13$};
}
\foreach \x in {8}
{\node[circle,draw=green,very thick,fill=green!30,inner sep=0.45em] at (14*\rozs,\x*\roz){};
\node at (14*\rozs,\x*\roz){$14$};
}
\end{tikzpicture} 
\caption{Structure of intervals  $(\ell,\ell+6]$.}\label{kulicky}
\end{center}
\end{figure}

It is easy to see that for each $\ell,u\geq 0$ and $b\in\{1,-1\}$
\begin{align}\label{2}
	\ep_{k,b}(\ell,\ell+2^u)=
	\left\{
	\begin{array}{cl}
	0 & \text{if $k\leq u-2$,}\\[1em]
	D_{k,b}(\ell,\ell+2^u) & \text{otherwise.}	
	\end{array}\right.
\end{align}

The whole problem concentrates in the distribution of those ``additional'' minus ones given by the mapping $\ep$. Let's therefore define
\begin{align*}\label{E}
	\E_{k,b}(s,d,m):=\left(\ep_{k,b}(s,s+d),\ep_{k,b}(s+d,s+2d),\dots,\ep_{k,b}(s+(m-1)d,s+md)\right).
\end{align*}
and
\[\Delta(s,d,m):=\sum_{k=0}^\infty \E_{k,b_k}(s,d,m). \]
Note that $\Delta(s,d,m)$ is just the vector \eqref{vektor} scaled down by a constant vector corresponding to the number of $-1$s expected given the length $d$. Our aim is to find (for a given $m$) numbers $s$ and $d$ such that $\Delta(s,d,m)$ is a constant vector.

\section{Additivity}
The key tool in the solution of the problem is the possibility to add two $\Delta$-vectors, formulated in the following lemma.
\begin{lemma}[Additivity of $\Delta$-vectors]\label{add}
Let $\sa$, $\sbb \geq 0$, and $\daa$, $\dbb \geq 1$ be positive integers such that $\sbb$ and $\dbb$ are even. Let $r$ be such that 
\begin{align}\label{r}
	2^{r}> \sa+m \daa, 
\end{align}
and
for each $i\geq 0$ the following implication holds:
\begin{align}\label{predpoklad}
	\text{if \quad $\E_{i,1}(\sbb,\dbb,m)\neq \E_{i,-1}(\sbb,\dbb,m)$\quad then	\quad $b_i=b_{i+r}$.}
	\end{align}
 Then
\[\Delta(\sa,\daa,m)+\Delta(\sbb,\dbb,m)=\Delta(\sa+2^r\sbb,\daa+2^r\dbb,m). \]
\end{lemma}
 \begin{proof}
Denote $\scc=\sa+2^r\sbb$ and $\dcc=\daa+2^r\dbb$.

Let firstly $k\leq r-1$. Since 
\begin{align*}
	\scc+j\dcc\equiv \sa+j \daa \mod 2^{k+2}
\end{align*}
holds for each $j$ (recall that $s'$ and $d'$ are even), we have  
\begin{align}\label{sum1}
	\E_{k,b_k}(\sa,\daa,m)=\E_{k,b_k}(\scc,\dcc,m).
\end{align}

Let now $k=r+i$ with $i\geq 0$. Let $j\in \{0,1,\dots,m-1\}$ and $b\in\{-1,1\}$. If
\[\dbb\gmod (2+b)\cdot 2^{i}-(\sbb+j\dbb)-1 \mod 2^{i+2},\]
then
\begin{align*}
	 \dbb \gmodeq (2+b)\cdot 2^{i}-(\sbb+j\dbb)\gmodeq 1 \mod 2^{i+2},
\end{align*}
which implies
\begin{align}\label{ineq1}
	 2^r\dbb \gmodeq (2+b)\cdot 2^{k}-2^r(\sbb+j\dbb)\gmodeq 2^r \mod 2^{k+2}.
\end{align}
From \eqref{r} we deduce (see Fact \ref{fact})
\[2^{r} \dbb+\daa\gmod(2+b)\cdot 2^{k}-2^{r}(\sbb+j\dbb)-\sa-j \daa- 1 \mod 2^{k+2}.\]
Indeed, both sides of the inequality \eqref{ineq1} are divisible by $2^r$ whence adding $\daa$ to the left side and subtracting $\sa+j\daa+1$ from the right side keeps both sides in the interval $\{0,1,\dots,2^{k+2}-1\}$.
We have shown  
\begin{align}\label{sum2}
	\ep_{i,b}(\sbb+j\dbb,\sbb+(j+1)\dbb)=\ep_{k,b}(\scc+j\dcc,\scc+(j+1)\dcc)=1.
\end{align}

On the other hand, if 
\[ (2+b)\cdot 2^{i}-(\sbb+j\dbb)-1 \gmodeq \dbb \mod 2^{i+2},\]
then
\begin{align}
	 (2+b)\cdot 2^{k}-2^{r}(\sbb+j\dbb)-2^{r} \gmodeq 2^{r}\dbb \mod 2^{k+2}. 
\end{align}
The inequality \eqref{r} implies $2^r-\sa-j\daa-1\geq \daa$, and, as above, we deduce
\[(2+b)\cdot 2^{k}-2^{r}(\sbb+j\dbb)-\sa-j\daa-1 \gmodeq 2^{r}\dbb+\daa \mod 2^{k+2}. \]
In this case, we have
\begin{align}\label{sum3}
	\ep_{i,b}(\sbb+(j-1)\dbb,\sbb+j\dbb)=\ep_{k,b}(\scc+(j-1)\dcc,\scc+j\dcc)=0.
\end{align}

 Using the assumption \eqref{predpoklad}, we deduce from \eqref{sum2} and \eqref{sum3}
\begin{align*}
	\E_{i,b_i}(\sbb,\dbb,m)=\E_{r+i,b_{r+i}}(\scc,\dcc,m).
\end{align*}
Since the inequality \eqref{r} implies $\E_{k,b_k}(\sa,\daa,m)=\nul$ for $k\geq r$, we have
\begin{align*}
	\Delta(\scc,\dcc,m)&=\sum_{k=0}^{r-1} \E_{k,b_k}(\scc,\dcc,m)+\sum_{i=0}^\infty \E_{r+i,b_{r+i}}(\scc,\dcc,m)=\\
                       &=\Delta(\sa,\daa,m) + \Delta(\sbb,\dbb,m),   	
\end{align*}
and the proof is complete.
\end{proof}

\section{Proof of the main claim}
Additivity of $\Delta$-vectors implies that in order to solve the problem it is enough to find $\Delta$\/-vectors that sum to a constant vector and the corresponding parts of the sequence of instructions are synchronized.

Next lemma indicates an interval which is free of $-1$s of high orders.
\begin{lemma}\label{interval}
For each $t\geq 0$ there is a number $\ell_t$ such that \[D_{k,b_{k}}(\ell_t,\ell_t+2^{t+2}-1)=0\] for all $k\geq t$.
\end{lemma}
\begin{proof} The number $\ell_t$ depends on values of $b_t$, $b_{t+1}$, $b_{t+2}$ and $b_{t+3}$ and we will give it explicitly.
Let $\ell=\ell(x_0,x_1,x_2,x_3)$, with $x_i\in \{1,-1\}$, be given by Figure \ref{hodnoty}. 
It can be directly checked that 
\[D_{0,x_0}(\ell,\ell+3)=D_{1,x_1}(\ell,\ell+3)=D_{2,x_2}(\ell,\ell+3)=D_{3,x_3}(\ell,\ell+3)=0\]
and also
\[D_{k,b}(\ell,\ell+3)=0\]
for both $b\in \{1,-1\}$ and each $k\geq 4$. Consequently, $\ell_0=\ell(b_0,b_1,b_2,b_3)$ has desired properties. Note that this is equivalent to $f_{\ell_0+1}=f_{\ell_0+2}=f_{\ell_0+3}=1$.

In general, we define \[\ell_t=2^t\ell(b_t,b_{t+1},b_{t+2},b_{t+3}).\]
Suppose that there is a number $n\equiv(2+b_k)\cdot 2^k \mod 2^{k+2}$ with $k\geq t$ and $\ell_t<n<\ell_t+2^{t+2}$.
Then $n'=2^{-t}n$ satisfies \[n'\equiv(2+b_k)\cdot 2^{k-t} \mod 2^{k-t+2}\] and
\[\ell(b_t,b_{t+1},b_{t+2},b_{t+3})<n'<\ell(b_t,b_{t+1},b_{t+2},b_{t+3})+4,\]
a contradiction with the first part of the proof.
\end{proof}

\begin{figure}
\centering 
\begin{tabular}{r r r r | c} 
\hline 
\ &\  &\ &\ &  \\[-2ex]
$x_0$ & $x_1$ & $x_2$ & $x_3$ & $\ell$ \\ [0.5ex] 
\hline 
\ &\  &\ &\ &  \\[-1.3ex] 
$1$ & $1$ & $1$ & $1$ & $7$ \\ 
$-1$ & $1$ & $1$ & $1$ & $1$\\
$1$ & $-1$ & $1$ & $1$ & $3$\\
$-1$ & $-1$ & $1$ & $1$ & $5$\\
$1$ & $1$ & $-1$ & $1$ & $7$\\
$-1$ & $1$ & $-1$ & $1$ & $9$\\
$1$ & $-1$ & $-1$ & $1$ & $11$\\
$-1$ & $-1$ & $-1$ & $1$ & $5$\\
$1$ & $1$ & $1$ & $-1$ & $23$\\
$-1$ & $1$ & $1$ & $-1$ & $1$\\
$1$ & $-1$ & $1$ & $-1$ & $3$\\
$-1$ & $-1$ & $1$ & $-1$ & $21$\\
$1$ & $1$ & $-1$ & $-1$ & $23$\\
$-1$ & $1$ & $-1$ & $-1$ & $9$\\
$1$ & $-1$ & $-1$ & $-1$ & $11$\\
$-1$ & $-1$ & $-1$ & $-1$ & $21$\\
\hline 
\end{tabular}
\caption{Values of $\ell(x_0,x_1,x_2,x_3)$.}
\label{hodnoty}
\end{figure}

\begin{remark}
	The interval $(\ell_t,\ell_t+2^{t+2}-1]$ from the previous lemma has length $2^{t+2}-1$ and contains no number equivalent to $(2+b_t)\cdot 2^t \mod 2^{t+2}$. Therefore, it spreads between two consecutive occurrences of such numbers. Moreover, the word $f_{\ell_t+1}f_{\ell_t+2}\cdots f_{\ell_t+2^{t+2}-1}$ has a period $2^{t+1}$, and it is of the form $w1w$. Lemma \ref{interval} is therefore a slightly stronger form of   \cite[Proposition 5]{narad}.  
\end{remark}

Lemma \ref{interval} has the following important consequence.

\begin{lemma}\label{Enula}
Let $t\geq 1$, $0\leq u \leq t$ and $0\leq p \leq 2^{t-u}-1$. Then
\begin{enumerate}
	\item \label{nula0} for each $k\geq 0$ such that $k\leq u-2$ or $k\geq t-1$, we have
	\[
\E_{k,b_k}(\ell_{t-1}+2^u p,2^u,2^{t-u})=\nul\,;
	\]
	\item \label{nula1} moreover, 
	\[
\E_{k,b}(\ell_{t-1}+2^u p,2^u,2^{t-u})=\nul
	\]
	for both $b\in\{-1,1\}$ and each $k\geq 0$ such that $k\leq u-2$ or $k\geq t+3$.
\end{enumerate}
\end{lemma}
\begin{proof} All intervals covered by the vector $\E_{k,b}(\ell_{t-1}+2^u p,2^u,2^{t-u})$ lie in $(\ell_{t-1},\ell_{t-1}+2^{t+1})$.
The claim now follows from \eqref{2}, Lemma \ref{interval} and from the fact that $\ell_{t-1}$ depends on values $b_{t-1}$, $b_{t}$, $b_{t+1}$ and $b_{t+2}$ only. 
\end{proof}

We now indicate vectors with a constant sum.
\begin{lemma}\label{konstant}
	For each $t\geq 2$ and $1\leq u < t$ 
	\[\vek:=\sum_{p=0}^{2^{t-u}-1}\Delta(\ell_{t-1}+2^u p,2^u,2^{t-u}),\]
is a constant vector.
\end{lemma}
\begin{proof}
From the definition of $\Delta$ and Lemma \ref{Enula}\eqref{nula0}, we deduce
\[
\vek=(\vek_0,\vek_1,\dots,\vek_{2^{t-u}-1})=\sum_{p=0}^{2^{t-u}-1}\sum_{k=u-1}^{t-2}\E_{k,b_k}(\ell_{t-1}+2^u p,2^u,2^{t-u}),
\]
and
\[
\vek_i=\sum_{p=0}^{2^{t-u}-1}\sum_{k=u-1}^{t-2}\ep_{k,b_k}(\ell_{t-1}+2^u(p+i),\ell_{t-1}+2^u(p+i+1)).
\]
By \eqref{2}, we have (see Figure \ref{obloucky})
\[
\vek_i=\sum_{k=u-1}^{t-2}D_{k,b_k}(\ell_{t-1}+2^u i,\ell_{t-1}+2^{t}+2^u i).
\]
Since the length of the interval $(\ell_{t-1}+2^u i,\ell_{t-1}+2^{t}+2^u i)$ is $2^{t}$, which is divisible by $2^{k+2}$ for all $k=\{0,1,\dots, t-2\}$, we finally obtain
\[
\vek_i=\sum_{k=u-1}^{t-2}\frac {2^{t}}{2^{k+2}}= 2^{t-u}-1,
\]
which is independent of $i$ as claimed.
\end{proof}
NB: the previous lemma holds also for $u=0$. We omit this case for sake of simplicity.
\begin{figure}
\begin{center}
\begin{tikzpicture}
\def\rad{1}
\def\jed{1}
\def\oblouk#1#2{\draw (#1*\jed,#2*\rad)..controls (#1*\jed+0.2*\jed,#2*\rad+0.4*\rad) and (#1*\jed+0.8*\jed,#2*\rad+0.4*\rad) ..(#1*\jed+1*\jed,#2*\rad);}
\def\vybarvi#1#2{\fill[green!50] (#1*\jed,#2*\rad)..controls (#1*\jed+0.2*\jed,#2*\rad+0.4*\rad) and (#1*\jed+0.8*\jed,#2*\rad+0.4*\rad) ..(#1*\jed+1*\jed,#2*\rad)--cycle;}
\foreach \y in {0,1,2,3}
{\vybarvi{2+\y}{-\y}
\foreach \x in {0,1,2,3}
{\oblouk{\x+\y}{-\y}}
}
\foreach \x in {0,1,2,3}
{\draw[shorten <=-5pt,shorten >=0.2pt, o-*] (0,-\x*\rad)--(4*\jed,-\x*\rad);\draw[shorten >=-5pt,-o](4*\jed,-\x*\rad)--(8*\jed,-\x*\rad);}
\node at (0,0.6*\rad){$\ell_{t-1}$};
\node at (4*\jed,0.6*\rad){$\ell_{t-1}+2^t$};
\node at (8*\jed,0.6*\rad){$\ell_{t-1}+2^{t+1}$};
\end{tikzpicture}
\end{center}
\caption{Illustration of Lemma \ref{konstant}. Green intervals sum to $\vek_2$.} \label{obloucky}
\end{figure}
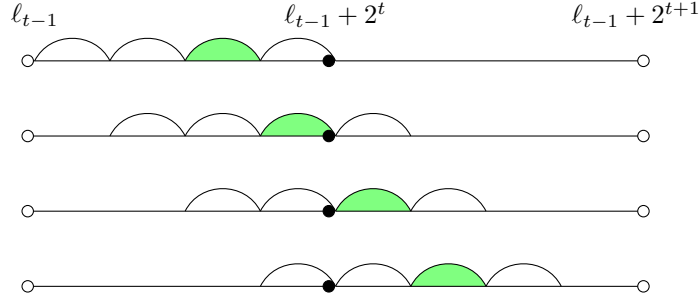

We are ready to prove the main claim of the paper.
\begin{theorem}\label{main}
All paper-folding words contain arbitrarily large abelian powers. 
\end{theorem}
\begin{proof}
Let $\f$ be a paper-folding word defined by a sequence of instructions $\bb$. Let $m=2^q$ for some $q\in \NN$. We shall find an abelian $m$-power in $\f$. Choose $u\geq 1$ and $t\geq 2$ such that $t-u=q$ and the factor $$b_{u-1}b_ub_{u+1}\cdots b_{t+1}b_{t+2}$$ occurs infinitely many times in $\bb$.
Using Lemma \ref{add} inductively, we construct numbers $s_{j}$ and $d_j$, with $j\in \{1,$ $2,\dots,$ $m\}$,  such that 
\[
\Delta(s_j,d_j,m)=\sum_{p=0}^{j-1}\Delta(\ell_{t-1}+2^u p,2^u,m).
\]
Clearly, $s_1=\ell_{t-1}$ and $d_1=2^u$. Numbers $s_{j+1}$, $d_{j+1}$ are defined by
\begin{align*}
	s_{j+1}&=s_j+2^{r_j}(\ell_{t-1}+2^uj), & d_{j+1}&=d_j+2^{r_j}2^u,
\end{align*}   
where $r_j$ is chosen to satisfy $2^{r_j}>s_j+md_j$ 
and
\[\text{$b_{i}=b_{i+r_j}$ for each $i=u-1,u,\dots,t+1,t+2$}.\] Since $u\geq 1$ and $t\geq 2$, all $d_j$ and $s_j$ are even. Lemma \ref{Enula}\eqref{nula1} implies that the assumption \eqref{predpoklad} of Lemma \ref{add} is fulfilled, whence
\[
\Delta(s_{j+1},d_{j+1},m)=\Delta(s_j,d_j,m)+\Delta(\ell_{t-1}+2^u j,2^u,m)=\sum_{p=0}^{j}\Delta(\ell_{t-1}+2^u p,2^u,m)
\]
as required. By Lemma \ref{konstant}, the vector $\Delta(s_m,d_m,m)$ is a constant vector, which means that 
\[f_{s_m+1}f_{s_m+2}\cdots f_{s_m+md_m},\]
is an abelian $m$-power.
\end{proof}

\section{Examples}
Let us use Theorem \ref{main} to find an abelian fourth power in the regular paper-folding sequence. Regularity allows to ignore the condition \eqref{predpoklad} of Lemma \ref{add}.  

Choose $u=1$ and $t=3$. We have $\ell_2=4\cdot \ell_0(1,1,1,1)=4\cdot 7=28$.
Lemma \ref{interval} claims that the word
 $f_{29}f_{30}\cdots f_{43}$ does not contain $-1$s of order higher than one. This, in particular, means that the word has period $8$.
\[
 \begin{tabular}{c | c c c c c c c c c c c c c c c c c } 
$i$ & $29$ & $30$ & $31$ & $32$ & $33$ & $34$ & $35$ & $36$ & $37$ & $38$ & $39$ & $40$ & $41$ & $42$ & $43$ \\ [0.5ex] \hline 
\\[-1em]
$f_i$ & $1$ & $-1$ & $-1$ & $1$ & $1$ & $1$ & $-1$ & $1$ & $1$ & $-1$ & $-1$ & $1$ & $1$ & $1$ & $-1$ \\ [0.5ex] 
\end{tabular}
\]
We want to sum vectors
\begin{align*}
		\Delta(28,2,4)&=(1,1,0,1),\\
		\Delta(30,2,4)&=(1,0,1,1),\\
		\Delta(32,2,4)&=(0,1,1,1),\\
		\Delta(34,2,4)&=(1,1,1,0).
\end{align*}
In Lemma \ref{add} we choose $r=r_1=6$ since $2^6>28+4\cdot 2=36$, and obtain
\[\Delta(28,2,4)+\Delta(30,2,4)=\Delta(28+2^6\cdot 30,2+2^6\cdot 2,4)=\Delta(1948,130,4)=(2,1,1,2).\]
 Since $2^{12}>1948+4\cdot 130>2^{11}$, we put $r_2=12$, and get
\begin{align*}
	\Delta(1948,130,4)+\Delta(32,2,4)&=\Delta(1948+2^{12}\cdot 32,130+2^{12}\cdot 2,4)=\\
	&=\Delta(133\,020,8322,4)=(2,2,2,3).
\end{align*}
Finally, with $r_3=18$, we have
\[\Delta(133\,020,8322,4)+\Delta(34,2,4)=\Delta(9\,045\,916,532\,610,4)=(3,3,3,3).\]
The following tables illustrate how the addition works.
\begin{align*}
&	 \begin{tabular}{c | c  }
$k$ & $\E_{k,1}(28,2,4)$ \\ [0.5ex] 
\hline 
\\[-1em]
$0$ & $(0,1,0,1)$ \\ 
$1$ & $(1,0,0,0)$  
\end{tabular}
&
&	 \begin{tabular}{c | c  }
$k$ & $\E_{k,1}(30,2,4)$ \\ [0.5ex] 
\hline 
\\[-1em]
$0$ & $(1,0,1,0)$ \\ 
$1$ & $(0,0,0,1)$  
\end{tabular}\\
&	 \begin{tabular}{c | c  }
$k$ & $\E_{k,1}(32,2,4)$ \\ [0.5ex] 
\hline 
\\[-1em]
$0$ & $(0,1,0,1)$ \\ 
$1$ & $(0,0,1,0)$  
\end{tabular}
&
&	 \begin{tabular}{c | c  }
$k$ & $\E_{k,1}(34,2,4)$ \\ [0.5ex] 
\hline 
\\[-1em]
$0$ & $(1,0,1,0)$ \\ 
$1$ & $(0,1,0,0)$  
\end{tabular}
\end{align*}
\[
\begin{tabular}{c | c  }
$k$ & $\E_{k,1}(9\,045\,916,532\,610,4)$ \\ [0.5ex] 
\hline 
\\[-1em]
$0$ & $(0,1,0,1)$ \\ 
$1$ & $(1,0,0,0)$ \\ 
$2$ & $(0,0,0,0)$ \\
$3$ & $(0,0,0,0)$ \\
$4$ & $(0,0,0,0)$ \\
$5$ & $(0,0,0,0)$ \\
$6$ & $(1,0,1,0)$ \\ 
$7$ & $(0,0,0,1)$ \\ 
$8$ & $(0,0,0,0)$ \\
$9$ & $(0,0,0,0)$ \\
$10$ & $(0,0,0,0)$ \\
$11$ & $(0,0,0,0)$ \\
$12$ & $(0,1,0,1)$ \\ 
$13$ & $(0,0,1,0)$ \\
$14$ & $(0,0,0,0)$ \\
$15$ & $(0,0,0,0)$ \\
$16$ & $(0,0,0,0)$ \\
$17$ & $(0,0,0,0)$ \\
$18$ & $(1,0,1,0)$ \\ 
$19$ & $(0,1,0,0)$ 
\end{tabular}
\]

Lemma \ref{add}, however, can be used to find more reasonable fourth power. One easily verifies that $\Delta(6,1,4)=(1,0,0,0)$ and $\Delta(0,2,4)=(0,1,1,1)$. Since $2^4>6+4\cdot 1$, we have
\[\Delta(6,1,4)+\Delta(0,2,4)=\Delta(6+2^4\cdot 0,1+2^4\cdot 2,4)=\Delta(6,33,4)=(1,1,1,1).\]
Corresponding tables are as follows.
\[
 \begin{tabular}{c | c  }
$k$ & $\E_{k,1}(6,1,4)$ \\ [0.5ex] 
\hline 
\\[-1em]
$0$ & $(1,0,0,0)$ \\ 
$1$ & $(0,0,0,0)$  
\end{tabular}
\quad\quad
 \begin{tabular}{c | c  }
$k$ & $\E_{k,1}(0,2,4)$ \\ [0.5ex] 
\hline 
\\[-1em]
$0$ & $(0,1,0,1)$ \\ 
$1$ & $(0,0,1,0)$  
\end{tabular}
\]
\[
\begin{tabular}{c | c  }
$k$ & $\E_{k,1}(6,33,4)$ \\ [0.5ex] 
\hline 
\\[-1em]
$0$ & $(1,0,0,0)$ \\ 
$1$ & $(0,0,0,0)$ \\
$2$ & $(0,0,0,0)$ \\ 
$3$ & $(0,0,0,0)$ \\
$4$ & $(0,1,0,1)$ \\ 
$5$ & $(0,0,1,0)$  
\end{tabular}
\]
Looking at the zero vectors for $k=1,2,3$ in the resulting table, one may be tempted to think that $2^4$ is unnecessarily large scaling ratio. However, this is not true, since we have
 \[\Delta(6,1,4)+\Delta(0,2,4)\neq\Delta(6+2^3\cdot 0,1+2^3\cdot 2,4)=\Delta(6,17,4)=(1,1,2,0).\]
\[
\begin{tabular}{c | c  }
$k$ & $\E_{k,1}(6,17,4)$ \\ [0.5ex] 
\hline 
\\[-1em]
$0$ & $(1,0,0,0)$ \\ 
$1$ & $(0,0,0,0)$ \\
$2$ & $(0,0,0,0)$ \\ 
$3$ & $(0,1,1,0)$ \\
$4$ & $(0,0,1,0)$ \\ 
\end{tabular}
\]
\section*{Acknowledgments}
I am grateful to James Currie, Narad Rampersad and, in particular, to Thomas Stoll for inspiring discussion about the problem and for useful comments. I also thank Jeffrey Shallit for his remarks.


\begin{thebibliography}{1}

\bibitem{Allouche1}
Jean-Paul Allouche.
\newblock The number of factors in a paperfolding sequence.
\newblock {\em Bull. Austral. Math. Soc.}, 46(1):23--32, 1992.

\bibitem{Allouche2}
Jean-Paul Allouche and Mireille Bousquet-M{\'e}lou.
\newblock Facteurs des suites de {R}udin-{S}hapiro g\'en\'eralis\'ees.
\newblock {\em Bull. Belg. Math. Soc. Simon Stevin}, 1(2):145--164, 1994.
\newblock Journ{\'e}es Montoises (Mons, 1992).

\bibitem{jeff}
Jean-Paul Allouche and Jeffrey O. Shallit.
\newblock {\em Automatic Sequences: Theory, Applications, Generalizations}.
\newblock Cambridge University Press, 2003.

\bibitem{Thue1906}
Jean Berstel.
\newblock Axel {T}hue's papers on repetitions in words: a translation.
\newblock Publications du LaCIM~20, Universit\'e du Qu\'ebec \`a Montr\'eal,
  1995.

\bibitem{Currie}
James~D. Currie and Terry~I. Visentin.
\newblock Long binary patterns are abelian 2-avoidable.
\newblock {\em Theoretical Computer Science}, 409(3):432 -- 437, 2008.

\bibitem{erdos}
Paul Erd{\H{o}}s.
\newblock Some unsolved problems.
\newblock {\em Magyar Tud. Akad. Mat. Kutat\'o Int. K\"ozl.}, 6:221--254, 1961.

\bibitem{pers}
Jeffrey Shallit.
\newblock Personal communication,
\newblock 2012.

\bibitem{narad}
Jui-Yi Kao, Narad Rampersad, Jeffrey Shallit, and Manuel Silva.
\newblock Words avoiding repetitions in arithmetic progressions.
\newblock {\em Theoret. Comput. Sci.}, 391(1-2):126--137, 2008.

\bibitem{veikko}
Veikko Ker{\"a}nen.
\newblock Abelian squares are avoidable on {$4$} letters.
\newblock In {\em Automata, languages and programming ({V}ienna, 1992)}, volume
  623 of {\em Lecture Notes in Comput. Sci.}, pages 41--52. Springer, Berlin,
  1992.

\end{thebibliography}

\end{document}